\documentclass[10pt, journal, compsoc]{IEEEtran}

\usepackage{booktabs}
\usepackage{graphicx}
\usepackage{balance}
\usepackage{color}
\usepackage{amssymb}
\usepackage{amsmath}
\usepackage{multirow}
\usepackage{booktabs}
\usepackage{makecell}
\usepackage[caption=false]{subfig}
\usepackage{xcolor, soul}
\usepackage{siunitx}
\usepackage{enumitem}
\usepackage[linesnumbered,vlined]{algorithm2e}

\usepackage[T1,hyphens]{url}
\usepackage{hyperref}

\usepackage{amsthm}

\theoremstyle{definition} 

\definecolor{DarkGray}{gray}{.45}
\newcommand{\return}{\vspace{0.3cm}}

\DeclareMathOperator*{\argmin}{arg\,min}

\newtheorem{defn}{Definition}
\newtheorem{thm}{Theorem}
\newtheorem{lem}{Lemma}

\newcommand{\enc}{\textsf{E}}
\newcommand{\bin}{\textsf{B}}
\newcommand{\curr}{\textsf{C}}
\newcommand{\last}{\textsf{L}}
\newcommand{\vb}{\textsf{VByte}}

\newcommand{\clue}{\textsf{ClueWeb09}}
\newcommand{\gov}{\textsf{Gov2}}
\newcommand{\cc}{\textsf{CCNews}}

\newcommand{\docs}{\textsf{docs}}
\newcommand{\freqs}{\textsf{freqs}}

\newcommand{\parag}[1]{{{\return\noindent\textbf{ #1.}}}}

\newcommand{\pp}{{\hspace{0.15cm}}}

\makeatletter
\newcommand{\removelatexerror}{\let\@latex@error\@gobble}
\makeatother

\SetKwIF{If}{ElseIf}{Else}{{\textsf{if}}}{}{\textsf{else if}}{{\textsf{else}}}{}
\SetKwFor{For}{\textsf{for}}{}{}
\newcommand{\code}[1]{{\textsf{\textbf{#1}}}}
\newcommand{\func}[1]{{\textsf{#1}}}
\newcommand{\var}[1]{{\textit{#1}}}

\begin{document}

\title{On Optimally Partitioning Variable-Byte Codes}


\author{Giulio Ermanno Pibiri and Rossano Venturini
\IEEEmembership{-- ISTI-CNR and University of Pisa, Italy}
\IEEEcompsocitemizethanks{\IEEEcompsocthanksitem
This work was partially supported by the BIGDATAGRAPES project
(grant agreement \#780751, European Union's Horizon 2020 research and innovation programme)
and MIUR-PRIN 2017 ``Algorithms, Data Structures and Combinatorics for Machine Learning''.
\protect\\
E-mails: \textsf{giulio.pibiri@di.unipi.it}, \textsf{rossano.venturini@unipi.it}.}
}


\IEEEtitleabstractindextext{%
\begin{abstract}
The ubiquitous \emph{Variable-Byte} encoding is one of the fastest compressed representation for integer sequences. However, its compression ratio is usually not competitive with other more sophisticated encoders, especially when the integers to be compressed are small that is the typical case for inverted indexes.
This paper shows that the compression ratio of Variable-Byte can be improved by $2\times$ by adopting a \emph{partitioned} representation of the inverted lists.
This makes Variable-Byte surprisingly competitive in space with the best bit-aligned encoders, \emph{hence} disproving the folklore belief that Variable-Byte is space-inefficient for inverted index compression.
Despite the significant space savings, we show that our optimization almost comes for free, given that:
we introduce an optimal partitioning algorithm that does not affect indexing time because of its linear-time
complexity;
we show that the query processing speed of Variable-Byte is preserved, with an extensive experimental analysis and comparison with several other state-of-the-art encoders.
\end{abstract}
}

\maketitle

\IEEEdisplaynontitleabstractindextext

\IEEEpeerreviewmaketitle

\section{Introduction}\label{sec:intro}
\IEEEPARstart{T}{he} \emph{inverted index} is the core data structure at the basis of large-scale search engines,
database architectures and social networks~\cite{2006:zobel.moffat,2008:manning.raghavan.ea,Buttcher-book,2013:curtiss.becker.ea,2012:busch.gade.ea}.
In its simplicity, the inverted index can be regarded as a collection of sorted integer sequences, called inverted or posting lists.
For example, when the inverted index is used to support full-text search in databases, each list is associated to a vocabulary term and stores the sequence of integer identifiers of the documents (docIDs) that contain such term~\cite{2008:manning.raghavan.ea}.
Then, identifying a set of documents containing all the terms in a user query reduces to the problem of intersecting the inverted lists associated to the terms in the query.
Likewise, an inverted list can be associated to a user in a social network (e.g., Facebook) and stores the sequence of all the friend identifiers of the user~\cite{2013:curtiss.becker.ea}.
Database systems based on SQL often precompute the list of row identifiers matching a specific frequent predicate over a large table, in order to speed up the execution of a query involving the conjunction of (possibly) many predicates~\cite{hristidis2003efficient,raman2007lazy}.
Also, finding all occurrences of twig patterns in XML databases can be done efficiently by resorting on an inverted index~\cite{bruno2002holistic}.
A common feature of key-value storage architectures like Apache Ignite and Redis,
is the organization of data elements falling into the same bucket due to a hash collision:
the list of all such elements is materialized that is, essentially, an inverted list~\cite{debnath2011skimpystash}.

Because of the huge quantity of data processed on a daily basis by the mentioned systems, \emph{compressing} the inverted index is indispensable since it can introduce a two-fold advantage over a non-compressed representation: feed faster memory levels with more data and, \emph{hence}, speed up the query processing algorithms.
As a result, the design of algorithms that compress the index effectively \emph{and} maintain a noticeable decoding speed is an old problem in Computer Science, that dates back to more than 50 years ago, and still a very active field of research.
Many representation for inverted lists are known, each exposing a different space/time trade-off~\cite{EBDT2018}.

Among these, \emph{Variable-Byte}~\cite{thiel1972program,1999:williams.zobel} (henceforth, $\vb$) is the most popular and used byte-aligned code.
In particular, $\vb$ owes its popularity to its \emph{sequential} decoding speed and, indeed, it is one of the fastest representation for integer sequences.
For this reason, it is widely adopted by well-known companies as a key tool to enable fast search of records.
We mention some noticeable examples.
Google uses $\vb$ extensively: for compressing the posting lists of inverted indexes~\cite{2009:dean} and as a binary wire format for its protocol buffers~\cite{protobuf}.
IBM DB2 employs $\vb$ to store the differences between successive record identifiers~\cite{bhattacharjee2009efficient}.
Amazon patented an encoding scheme, based on $\vb$ and called Varint-G8IU, which uses SIMD (Single Instruction Multiple Data) instructions to perform decoding faster~\cite{2011:stepanov.gangolli.ea}.
Many other storage architectures rely on $\vb$, such as Redis~\cite{redisearch}, UpscaleDB~\cite{upscaledb} and Dropbox~\cite{dropboxtechblog}.

\begin{figure*}
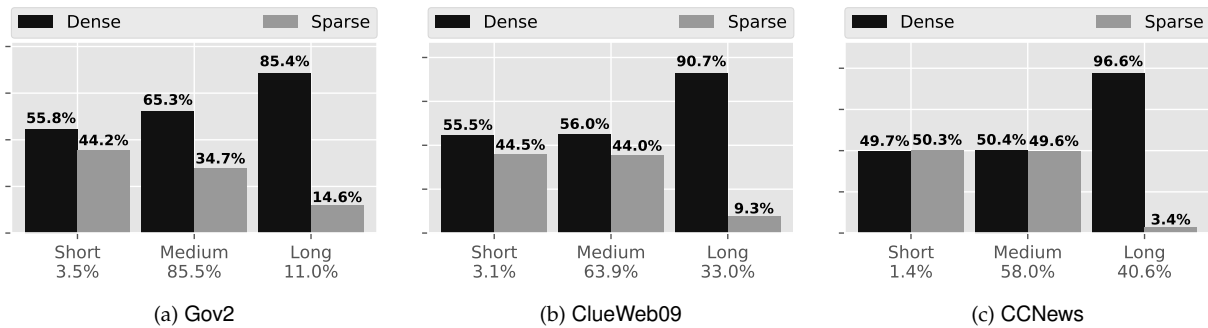

    \centering
    \subfloat[\gov]{
    \includegraphics[scale=0.65]{{{plots/gov2.postings_distribution.10000.7000000.bw}}}
    }
    \subfloat[\clue]{
    \includegraphics[scale=0.65]{{{plots/clueweb09.postings_distribution.10000.7000000.bw}}}
    }
    \subfloat[\cc]{
    \includegraphics[scale=0.65]{{{plots/ccnews.postings_distribution.10000.7000000.bw}}}
    }
    \caption{Percentage of integers belonging to \emph{dense} and \emph{sparse} regions of the inverted lists for the tested datasets. The inverted lists have been clustered by size into three categories: short (size $<$ 10K); medium (10K $\leq$ size $<$ 7M); long (size $\geq$ 7M). Below each category we also indicate the percentage of integers belonging to its inverted lists.}
     \label{fig:postings_distr}
\end{figure*}

We now quickly review how the $\vb$ encoding works.
It was first described by Thiel and Heaps~\cite{thiel1972program}.
The binary representation of a non-negative integer is divided into groups of $7$ bits which are represented as a sequence of bytes. In particular, the $7$ least significant bits of each byte are reserved for the data whereas the most significant (the $8$-th), called the \emph{continuation bit}, is equal to 1 to signal continuation of the byte sequence.
The last byte of the sequence has its $8$-th bit set to 0 to signal, instead, the termination of the byte sequence.
As an example, the integer \num{65 790} is represented as $\underline{1}0000100{\,}\underline{1}0000001{\,}\underline{0}1111110$ (with control bits underlined.)
Also, notice the padding bits, in the first byte starting from the left, inserted to align the binary representation of the number to a multiple of $8$ bits.
In particular, $\vb$ uses $\lceil \frac{\lceil \log_2 (x + 1) \rceil}{7} \rceil \times 8$ bits to represent an integer $x \geq 0$.
Decoding is simple: we just need to read one byte at a time until we find a value smaller than $2^7$.
As already mentioned, the format is also suitable for SIMD instructions for speeding up sequential decoding.

The main drawback of $\vb$ lies in its byte-aligned nature, which means that the number of bits needed to encode an integer cannot be less than $8$. For this reason, $\vb$ is only suitable for large numbers.
However, the inverted lists are notably known to exhibit a \emph{clustering effect}, i.e., these contain regions of (very) close identifiers that are far more compressible than highly scattered regions~\cite{EBDT2018,2000:moffat.stuiver,2014:ottaviano.venturini}.
Such natural clusters are present because the indexed data itself tend to be very similar.
As a simple example, consider all the Web pages belonging to the same domain: these are likely to share a lot of terms.
Also, the values stored in the columns of databases typically exhibit high locality and that is why column-oriented databases can achieve very good compression and high query throughput~\cite{abadi2013design}.

The key point is that effective inverted index compression should exploit as much as possible the clustering effect of the inverted lists.
$\vb$ currently fails to do so and, as a consequence, it is believed to be space-inefficient for inverted indexes, because its space occupancy can be up to $3\times$ larger than bit-aligned compressors~\cite{2014:ottaviano.venturini,2015:ottaviano.tonellotto.ea,2017:pibiri.venturini,ANS2}.

\parag{The motivating experiment}
As an illustrative example, consider the following two sequences: $\langle 1,2,3,4,5 \rangle$ and $\langle 127,254,318,408,533 \rangle$.
To reduce the values of the integers, $\vb$ compresses the differences between successive values, known as \emph{delta}-gaps or $d$-gaps, i.e., the sequences $\langle 1,1,1,1,1 \rangle$ and $\langle 127, 127,64,90,125 \rangle$ respectively (the first integer is left as it is).
Now, it is easy to see that $\vb$ will use 5 bytes to encode \emph{both} sequences,
but the first one can be compressed much better, with just $\approx$log$_2$ 5 bits.
To better highlight how this behavior can deeply affect compression effectiveness, we consider the statistic shown in Fig.~\ref{fig:postings_distr}.
This statistic reports the percentage of integers belonging to \emph{dense} and \emph{sparse} regions of the lists for the datasets $\gov$, $\clue$ and $\cc$.
More precisely, the plot originated from the following experiment: we divided each inverted list into blocks of 128 integers and we considered as \emph{sparse} a block where $\vb$ yielded a better space occupancy with respect to the \emph{characteristic bit-vector} representation of the block (if $u$ is the last element in the block, we have the $i$-th bit set in a bitmap of size $u$ for all integers $i$ belonging to the block), regarded to as the \emph{dense} case.
We also clustered the inverted lists by their sizes, in order to show where dense and sparse regions are most likely to be present.

The experiment clearly shows that \emph{we have a majority of dense regions}, thus explaining why in this case $\vb$ is not competitive with respect to bit-aligned encoders and, thus, motivating the need for introducing a better encoding strategy that adapts to such distribution.
We can also conclude that such optimization is likely to pay off because almost the entirety of integers concentrate in the lists of medium and long size (thanks to the Zipfian distribution of words in text), where indeed the majority of them belong to dense chunks.

\parag{Our contributions} We list here our main contributions. \\\\
\noindent
\textbf{1)}
We disprove the folklore belief that $\vb$ is too large to be considered space-efficient for representing inverted indexes, by exhibiting an improved compression ratio of $2\times$ on standard datasets consisting in several billions of integers, such as {\gov}, {\clue} and {\cc}.

The result is achieved by partitioning the inverted lists into blocks and representing each block with the most suitable encoder, chosen among $\vb$ and its characteristic bit-vector representation.
Partitioning the lists has the potential of adapting compression to the distribution of the integers in the lists, such as the one shown in Fig.~\ref{fig:postings_distr}, i.e. using $\vb$ for the sparse regions where larger $d$-gaps are likely to be present. \\

\noindent
\textbf{2)}
Since we cannot expect the dense regions of the lists be always aligned with uniform boundaries, we consider the problem of minimizing the space of representation of an inverted list of size $n$ by representing it with variable-length partitions.
By exploiting the fact that \textsf{VByte} is a \emph{point-wise} encoder, i.e.,
the number of bits to represent an integer solely depends on the \emph{value} of the integer itself and not on the universe and size of the block to which it belongs to, we introduce an algorithm that finds an \emph{optimal} partitioning in $\Theta(n)$ time and $O(1)$ space.

We remark that an algorithm based on dynamic programming~\cite{2014:ottaviano.venturini} can be used as well to find a $(1+\epsilon)$-optimal solution to the problem, by taking $O(n \log_{1+\epsilon}\frac{1}{\epsilon})$ time and $O(n)$ space for any $\epsilon \in (0,1)$.
Apart from being approximated rather than exact, this solution is generally applicable to any encoder whose size in bits can be computed (or estimated) in constant time on a block and does not rely on the fact that \textsf{VByte}
is point-wise. As a consequence, it is also noticeably slower than the algorithm we introduce in this work.
Lastly, we also remark that, although we use \textsf{VByte} in the experiments, our approach can be applied to \emph{any} point-wise encoder. \\

\noindent
\textbf{3)}
We conduct an extensive experimental analysis to demonstrate the effectiveness of our approach on standard large datasets, such as $\gov$, $\clue$ and $\cc$.
More precisely, when compared to the un-partitioned $\vb$ indexes, the optimally-partitioned counterparts are: (1) \emph{significantly smaller}, by $2\times$ on average; (2) \emph{not} slower at computing boolean conjunctions; (3) even faster to build on large datasets thanks to the introduced fast partitioning algorithm and improved compression ratio.

We compare the performance of partitioned $\vb$ indexes against several state-of-the-art encoders, such as: partitioned Elias-Fano (\textsf{PEF})~\cite{2014:ottaviano.venturini}, Binary Interpolative coding (\textsf{BIC})~\cite{2000:moffat.stuiver}, the optimized PForDelta (\textsf{OptPFD})~\cite{2009:yan.ding.ea},
an index organization~\cite{ANS2} based on Asymmetric Numeral Systems (\textsf{ANS})~\cite{duda2013asymmetric} and the \textsf{QMX} mechanism~\cite{2014:trotman}.
The partitioned $\vb$ representation reduces the gap between the space of $\vb$ and the one of the best bit-aligned compressors, such as \textsf{PEF} and \textsf{BIC},
by passing from an average original gap of $138\%$ to only $11\%$ with respect to \textsf{PEF}; from $174\%$ to only $22\%$ with respect to \textsf{BIC}.
Moreover, it does not introduce any query processing overhead: only \textsf{QMX} is slightly faster on {\clue} by $1 \div 10\%$, but also $30\%$ larger.

\section{Problem statement and related work on partitioning algorithms}\label{sec:related}

In this paper we study the problem of partitioning a sorted integer sequence $S$ of size $n$ to improve its compression, by adopting a \emph{2-level} representation.
This representation stores $S$ as a sequence of partitions $L_2[S_1,\ldots,S_k]$ that are concatenated in the second level $L_2$. The first level $L_1$ stores, instead, a fixed amount of bits, say $F$, for each partition $S_i$, needed to describe its size $n_i$ and largest element $u_i$. Clearly, $F$ can be safely upper bounded by $O(\log u)$ bits.
This representation has several important advantages over a shallow representation:
\begin{enumerate}
\item it permits to choose the most suitable encoder for each partition, given its size and upper bound, hence improving the overall space;
\item each partition $S_i$ can be represented in a smaller universe, i.e., $u_i - u_{i-1} - 1$, by subtracting to all its elements the ``base'' value $u_{i-1} + 1$, thus contributing to further reduction in space;
\item it allows a faster access to the individual elements of $S$, since we can first locate the partition to which an element belongs to and, then, conclude the search in that partition only.
\end{enumerate}

\parag{The problem}
Now, the natural arising problem is \emph{how} to choose the lengths and encoders for each partition in order to minimize the space of $S$.
As already noted, the problem is not trivial since we cannot expect dense regions of the lists being always aligned with fix-sized partitions.
While a dynamic programming recurrence computes an optimal solution to this problem in $\Theta(n^2)$ time and $O(n)$ space by (trivially) considering the cost of all possible splittings, this approach is clearly unfeasible already for modest sizes of the input.
Therefore, we need smarter methods such as the ones we describe in the following.

\subsection{Partitioning algorithms}\label{subsec:opt_part}
The simplest partitioning strategy chooses the length $b$ of every partition, e.g., $b = 128$ integers, and splits the list into $\lceil n / b \rceil$ blocks (the last partition can be possibly smaller than $b$ integers).
We call this partitioning strategy, \emph{uniform}.
The advantage of this representation is simplicity, since no expensive calculation is needed prior to encoding.
However, we cannot expect this strategy to yield the most compact indexes because the highly clustered regions of inverted lists could likely be broken by such fix-sized partitions.

This is the main motivation for introducing optimization algorithms that try to find the best partitioning of the list, thus minimizing its space of representation.
Silvestri and Venturini~\cite{2010:silvestri.venturini} obtain a $O(h \times n)$ construction time, where $h$ is the size of the longest allowed partition.
Ferragina \emph{et al.}~\cite{2011:ferragina.nitto.ea} improve the result by Buchsbaum \emph{et al.}~\cite{2003:buchsbaum.ea} by computing a partitioning whose cost is guaranteed to be at most $(1+\epsilon)$ times away from the optimal one, for any $\epsilon \in (0,1)$, in $O(n \log_{1+\epsilon} n)$ time.
Their approach can be applied to any encoder $\enc$ whose cost in bits can be computed (or, at least, estimated) in constant time for any portion of the input.

\parag{Dynamic programming: slow and approximated}
Ottaviano and Venturini~\cite{2014:ottaviano.venturini} resort to similar ideas to the ones presented by Ferragina \emph{et al.}~\cite{2011:ferragina.nitto.ea} to obtain a running time of $O(n \log_{1+\epsilon}\frac{1}{\epsilon})$, and yet, preserving the same approximation guarantees.
Note that the complexity is $\Theta(n)$ as soon as $\epsilon$ is constant.
The core idea of this approach is to not consider \emph{all} possible splittings, but only the ones whose cost is able of amortizing the fix cost $F$.
We quickly review their approach.

The problem of determining the partitioning of minimum cost can be modeled as the problem of finding a path of minimum cost (shortest) in a complete, weighted and directed acyclic graph (DAG) $\mathcal{G}$. This DAG has $n$ vertices, one for each position of $S$, and it is complete, i.e., it has $\Theta(n^2)$ edges where the cost $C(i,j)$ of edge $(i,j)$ represents the number of bits needed to represent $S[i, j]$.
Since the DAG is complete, a simple shortest path algorithm will not suffice to compute an optimal solution efficiently.
Thus, we proceed by \emph{sparsification} of $\mathcal{G}$, as follows.
We first consider a new DAG $\mathcal{G}_\epsilon$, which is obtained from $\mathcal{G}$ and has the following properties: (1) the number of edges is $O(n \log_{1+\epsilon}\frac{U}{F})$ for any given $\epsilon \in (0,1)$; (2) its shortest path distance is at most $(1+\epsilon)$ times the one of the original DAG $\mathcal{G}$, where $U$ represents the encoding cost of $S$ when no partitioning is performed.
It can be proven that the shortest path algorithm on $\mathcal{G}_\epsilon$ finds a solution which is at most $(1+\epsilon)$ times larger than an optimal one, in time $O(n \log_{1+\epsilon}\frac{U}{F})$, because $\mathcal{G}_\epsilon$ has $O(n \log_{1+\epsilon}\frac{U}{F})$ edges~\cite{2011:ferragina.nitto.ea}.
To further reduce the complexity by preserving the same approximation guarantees, we define two approximation parameters: $\epsilon_1 \in [0,1)$ and $\epsilon_2 \in [0,1)$.
We first retain from $\mathcal{G}$ all the edges whose cost is no more than $L = \frac{F}{\epsilon_1}$, then we apply the pruning strategy described above with $\epsilon_2$ as approximation parameter. The obtained graph has now $O(n \log_{1+\epsilon_2}\frac{L}{F}) = O(n \log_{1+\epsilon_2}\frac{1}{\epsilon_1})$ edges, which is $\Theta(n)$ as soon as $\epsilon_1$ and $\epsilon_2$ are constant.
Again, it can be proven that the shortest path distance is no more than $(1+\epsilon_1)(1+\epsilon_2) \leq (1+\epsilon)$ times the one in $\mathcal{G}$ by setting $\epsilon_1 = \epsilon_2 = \frac{\epsilon}{3}$~\cite{2014:ottaviano.venturini}.

Despite the \emph{theoretical} linear-time complexity for a constant $\epsilon$, the main drawback of the algorithm lies in the high constant factor.
For example, even by setting $\epsilon = 0.03$
we obtain a hidden constant of $\log_{1+0.03} \frac{1}{0.03} \simeq 118.63$, which results in a noticeable cost in practice.
Although enlarging $\epsilon$ can reduce the constant at the price of reducing the compression efficacy, this remains the bottleneck for the building step of large inverted indexes.

\section{Optimal partitioning in linear time: fast and exact}\label{sec:optimal_splitting}
The interesting research question we now pose is whether there exist an algorithm that finds an \emph{exact} solution, rather than approximated, in linear time and with \emph{low} constant factors.
This section answers positively to this question by showing that if the cost function of the chosen encoder is \emph{point-wise}, i.e., the number of bits need to represent an integer solely depends on the value of such integer and not on the universe and size of the partition it belongs to, the problem of determining and optimal partition admits an \emph{exact} solution in $\Theta(n)$ time and $O(1)$ space.

In the following, we first overview and discuss our solution by explaining the intuition that lies at its core, then we give the full technical details along with a proof of optimality and the relative pseudocode.

\subsection{Overview}
We are interested in computing the partitioning of $S$ whose encoding cost is minimum by using \emph{two} different encoders that take into account the relation between the size and universe of each partition. We already motivated the potential of this strategy by commenting on Fig.~\ref{fig:postings_distr}, which shows the distribution of the integers in dense and sparse regions of the inverted lists.
Let us consider the partition $S[i,j)$, $0 \leq i < j \leq n$, of relative universe $u = S[j-1]-S[i-1]-1$ and size $b = j - i$. Intuitively, when $b$ gets closer to $u$ the partition becomes denser; vice versa, it becomes sparser whenever $b$ diverges from $u$.
Thus the encoding cost $C(S[i,j))$ is chosen to be the minimum between $\bin(S[i,j)) = u$ bits (dense case) and $\enc(S[i,j))$ bits (sparse case), where $\bin$ is the characteristic bit-vector representation of $S[i,j)$ and $\enc$ is the chosen point-wise encoder for sparse regions.

Examples of point-wise encoders are $\vb$~\cite{thiel1972program}, Elias' $\gamma$-$\delta$~\cite{1975:elias} and Golomb~\cite{1966:golomb}.
Other encoders, such as Elias-Fano~\cite{1971:fano,1974:elias}, Binary Interpolative coding~\cite{2000:moffat.stuiver} and PFor-Delta~\cite{2009:yan.ding.ea} are not point-wise, since a different number of bits could be used to represent the \emph{same} integer when belonging to partitions having different characteristics, namely different size and universe.
To clarify what we mean, consider the following example sequence:
$$S[0,10) = \langle 8, 9, 10, 11, 12, 36, 37, 38, 39, 40 \rangle.$$
Let us now compare the behavior of Elias-Fano (non point-wise) and $\vb$ (point-wise).
By performing no splitting, Elias-Fano will use $\lceil \log_2(40/10)\rceil + 2 = 4$ bits to represent each integer.
By performing the splitting $[0,5)[5,10)$, the first five values will be represented with $4$ bits each, but the next five values with $\lceil \log_2((40 - 12 - 1) / 5) \rceil + 2 = 5$ bits each. Instead, by performing the splitting $[0,6)[6,10)$, the first six values will use $5$ bits each, while the next four only $2$ bits each.
Thus, performing different splittings changes the cost of representation of the \emph{same} postings for a non point-wise encoder, such as Elias-Fano.
Instead, it is immediate to see that $\vb$ will encode each element with 8 bits, regardless of any partitioning.

\parag{The intuition}
The above example gives us an intuitive explanation of why it is possible to design a light-weight approach for a point-wise encoder $\enc$:
we can compute the number of bits needed to represent a partition of $S$ with $\enc$ by just scanning its elements and summing up their costs, \emph{knowing that performing a splitting will not change their cost of representation} nor, therefore, the one of the partition.
This means that as long as the cost $\enc(S[0,j))$, for some $0 < j \leq n$, is less than $\bin(S[0,j))$ we know that $S[0,j)$ will be better represented with $\enc$ rather than with $\bin$.
Therefore, we can safely keep scanning the sequence until the difference in cost between $\enc(S[0,j))$ and $\bin(S[0,j))$ becomes more than $F$ bits. At this point, it means that $\enc$ is wasting more than $F$ bits with respect to $\bin$, thus we should stop encoding with $\enc$ the current partition because we can afford to pay the fix cost $F$ and continue the encoding with $\bin$.
Now, the crucial question is: at which position $k < j$ should we stop encoding with \textsf{E} and switch to \textsf{B}?
The answer is simple: we should stop at the position $k < j$ at which we saw the \emph{maximum} difference between the costs of $\enc$ and $\bin$, because splitting in \emph{any} other point will yield a larger encoding cost.
In other words, $k$ represents the position at which $\enc$ \emph{gains the most} with respect to $\bin$, so we will be wasting bits by splitting before or after position $k$.
Observe that we must also require such gain be more than $F$ bits,
otherwise switching encoder actually causes a waste of bits. In other terms, we say that in such case the gain would not be sufficient to amortize the fixed cost of the partition, meaning that we should \emph{not} split the sequence yet.

In conclusion, we encode $S[0,k)$ with $\enc$ and know that the elements $S[k,j)$ will now be better represented with $\bin$, rather than with $\enc$.
Fig.~\ref{fig:dominating} offers a pictorial representation of how the difference between the encoding costs of $\enc$ and $\bin$, referred to as the \emph{gain} function, changes during the scan of $S$. When the function is decreasing, it means that $\enc$ is winning over $\bin$, i.e., its encoding cost is less; conversely, when $\bin$ is more effective than $\enc$, the function is increasing.

After encoding the first partition $S[0,k)$, the process repeats: (1) we keep scanning $S$ until $\bin$ loses more than $2F$ bits with respect to $\enc$; at that point (2) we encode with $\bin$ the elements in $S[k, k^\prime)$ if the maximum gain of $\bin$ with respect to $\enc$, seen at position $k^\prime$, is greater than $2F$ bits.
We keep alternating compressors until the end of the sequence.

\begin{figure}
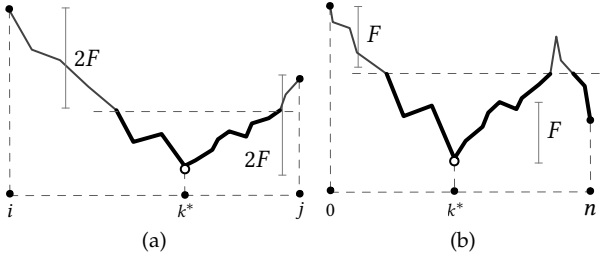

    \centering
    \subfloat[]{
    \includegraphics[scale=0.75]{{{imgs/gain1}}}
    \label{fig:dominating1}
    }
    \subfloat[]{
    \includegraphics[scale=0.75]{{{imgs/gain2}}}
    \label{fig:dominating2}
    }
    \caption{In case (a), we should split $[i,j)$ in $k^{\ast}$ because there the gain is the minimum among all points whose gain is below $2F$ from $i$ \emph{and} $j$;
    in case (b) we should \emph{not} split the sequence because, although an increase in gain of $F$ bits follows, we do not have a sufficiently high gain up to position $n$ to amortize the cost of the splitting.
    \label{fig:dominating}}
\end{figure}

\return
Before sketching a compact pseudocode of our algorithm, we first express some considerations.
First of all note that, for all partitions except the first, we need to amortize twice the fix cost, because we could potentially merge the last formed partition with the current one, thus, in order to be beneficial, the difference in the cost of the two encoders must be larger than $2F$ bits.
Again, refer to Fig.~\ref{fig:dominating1} for an example.
Also, for illustrative purposes, in the above discussion we have assumed that the first partition is encoded with $\enc$: clearly, $\bin$ could be better at the beginning but the algorithm will work in the very same way.

\parag{The algorithm}
In the most general terms, call $\last$ the encoder used to represent the \emph{last} encoded partition and $\curr$ the \emph{current} one.
These will be either $\enc$ or $\bin$.
We also indicate with the same letters the costs in bits of their representation of the current partition.
Finally, let $g^*$ indicate the best gain of $\curr$ with respect to $\last$.
At a high level, the skeleton of our algorithm looks as follows.
\begin{enumerate}
\item Encode the first partition.
\item Until the end of the sequence: if $| \curr - \last |$ and $g^*$ are greater than $2F$ bits, encode the current partition with $\curr$ and swap the roles of $\curr$ and $\last$.
\item Encode the last partition.
\end{enumerate}
In the above pseudocode, the encoding of the first and last partitions its treated separately because these must amortize a fix cost of $F$ bits instead of $2F$ bits: in fact, we do not have any partition before and after, respectively (see Fig.~\ref{fig:dominating2}).

It is immediate to see that the described approach can be implemented by using $O(1)$ space because we only need to keep the difference between the costs of $\enc$ and $\bin$ (plus some cursor variables), and that it runs in $\Theta(n)$ time because we calculate the cost in bits of each integer exactly once.
We have, therefore, eliminated the linear-space complexity of \emph{any} dynamic programming approach because we do not need to maintain the costs of the shortest path ending in each position of $S$.
Moreover, the introduced algorithm has very low constant factors in the time complexity, since it just performs few comparisons and updates of some variables for each integer of $S$.

\subsection{Technical discussion}
Let $S$ be a sorted integer sequence of size $n$.
In order to describe the properties of our solution, we first need the following definitions.

\begin{defn}\label{defn:gain}
Let $g : \mathbb{N} \cup \{0\} \rightarrow \mathbb{Z}$ be the \emph{gain} function, defined as
\begin{equation}
g(i) = 
     \begin{cases}
        0 & i = 0 \\
       \sum_{k = 0}^{i-1}[\enc(S[k]) - \bin(S[k])] & 0 < i \leq n
     \end{cases}.
\end{equation}
\end{defn}

\begin{defn}\label{defn:dominating_point}
Given the interval $[i,j)$ with $0 \leq i < j \leq n$, the position $k^{\ast}$ is \emph{dominating} $[i,j)$ for the encoder $\enc$, if
\begin{equation}\label{eq:dominating_point}
k^{\ast} = \argmin_{i < k \leq j} g(k) \mbox{ such that } g(i) - g(k^{\ast}) > T,
\end{equation}
where $T = F$ if $i = 0$ or $2F$ otherwise, and $j$ satisfies one of the following:
\begin{align}
& g(j) - g(k^{\ast}) > 2F, \mbox{ or} \label{eq:cond_1} \\
& g(z) - g(k^{\ast}) > F, \mbox{ for all } z \geq j. \label{eq:cond_2}
\end{align}
\end{defn}

\noindent
Notice that the dominating point could not exist for any interval $[i,j)$, but if it exists and $\enc(x) \neq \bin(x)$ for any $x \in S[i,j)$, it must be unique.
The definition of dominating point for encoder $\bin$ is symmetric to Definition~\ref{defn:dominating_point}.

The definition of dominating point explains that we can \emph{always improve} the cost of representation of $S[i,j)$ by splitting $[i,j)$ in the dominating point if it exists, otherwise we should \emph{not} split $[i,j)$.
It is easy to see that the point dominating $[i,j)$ is the point in which the difference of the costs between the two compressors is maximized, thus it will be only beneficial to split in this point rather than any other point, as we explained in the previous paragraph.
It is also easy to see why we should search the point dominating $[i,j)$ among the ones whose gain is at least $T$ bits less than $g(i)$.
The threshold $T$ is set to the minimum amount of bits needed to amortize the cost of switching from one compressor to the other.
Consider Fig.~\ref{fig:dominating1} and suppose we are encoding with $\bin$ before position $i$ and after $k^{\ast}$.
If we compress with $\enc$ the partition $S[i, k^{\ast})$, we are switching encoder twice, thus the gain in $k^{\ast}$ must be at least $2F$ bits less than $g(i)$ to be able of amortizing the cost for two switches. In Fig.~\ref{fig:dominating2}, instead, we have no partition before position $0$, thus we strive to amortize the cost for a single switch.

\return
\emph{Our strategy consists in splitting the sequence in the dominating points}.
More precisely, the solution $\mathcal{P} = [p_1, \ldots, p_m]$, $m \geq 1$, output by this strategy can be described by the following recursive equation.
\begin{equation}\label{eq:recursive}
p_i = 
     \begin{cases}
        0 & i = 0 \\
        n & i = m \\
        \text{dominating } [p_{i-1}, p_{i+1}) & \text{otherwise}
     \end{cases}
\end{equation}
In other words, any position in $\mathcal{P}$, except for the first and the last, is the dominating point of the interval whose endpoints are dominating points as well.

Notice that, by definition, there cannot be two adjacent dominating points that are relative to the same encoder, but they must be relative to different encoders.
In fact, suppose we have a dominating point $k^{\ast}$ for $\enc$. It means that we have seen an increase in gain of $2F$ after $k^{\ast}$, therefore: either the gain will then decrease sufficiently to find a dominating point for $\bin$, or the gain will never does so, thus, a dominating point after $k^{\ast}$ if not found.
This means that $\mathcal{P}$ alternates the choice of compressors, i.e., a partition encoded with $\enc$ is delimited by two partitions encoded with $\bin$ and viceversa (except for the first and last).
We call such behavior, \emph{alternating}.

In particular, our strategy will encode with compressor $\enc$ all partitions ending with a dominating point for $\enc$ (and starting with a dominating point for $\bin$, since $\mathcal{P}$ alternates the compressors). Symmetrically, the same holds for $\bin$.
As already pointed out, the only exception is made for the last partition, because the position $n$ cannot be dominating by definition (no increase or decrease in gain is possible after the end of the sequence).
In this case, the strategy selects the compressor that yields the minimum cost over the last partition.

\return
Since a \emph{feasible solution} to the problem is just either a singleton partition or consists in any sequence of strictly increasing positions, we argue that $\mathcal{P}$ is a feasible solution. This follows automatically by the definition of dominating point because such points are different from each other and, therefore, strictly increasing.
If no dominating points exist, then $\mathcal{P}$ will only contain position $n$ (one-past the end): it is a feasible solution too and indicates that $S$ should not be cut (singleton partition).
We now show the following lemma.

\begin{figure}
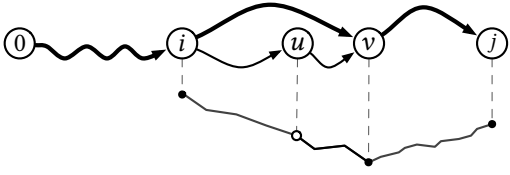

    \centering
    \includegraphics[scale=0.8]{{{imgs/proof}}}
    \caption{Path of minimum cost till position $j$ (thick black line) and its representation in terms of gain function; $u$ is the point dominating $[i,j)$.}
     \label{fig:proof}
\end{figure}

\begin{lem}\label{lem:optimality}
$\mathcal{P}$ is \emph{optimal}.
\end{lem}
\begin{proof}
As already noted in Section~\ref{subsec:opt_part}, an optimal solution to the problem can be thought as a path of minimum cost in the DAG whose vertices are the positions of the integers of $S$ and $C(i,j)=C(S[i,j])$ for any edge $(i,j)$.
Thus, suppose that $\mathcal{P}$ is not a shortest path and let $\mathcal{P}^* = [p^*_1, \ldots, p^*_m]$ be the shortest path sharing the longest common prefix with $\mathcal{P}$.
Refer to Fig.~\ref{fig:proof} for a graphical representation: $i$ is the largest position shared by $\mathcal{P}^*$ and $\mathcal{P}$.
We want to show that we can replace the edge $(i, v)$, $v \in \mathcal{P}^*$, with the path $(i, u)(u, v)$, $u \in \mathcal{P}$, without changing the cost of $\mathcal{P}^*$, therefore extending the longest common prefix up to node $u < v$ (the case for $u > v$ is symmetric).
We argue that this is only possible if $\mathcal{P}$ is optimal, otherwise it would mean that $\mathcal{P}^*$ is not a shortest path sharing a common longest prefix with $\mathcal{P}$, which is absurd by assumption.


First note that both edges $(i, v)$ and $(i, u)$ must be encoded with the same compressor.
In fact, suppose that these are not, for example $(i, v)$ is encoded with $\bin$ and $(i, u)$ with $\enc$.
Since $v \in \mathcal{P}^*$, we know that it is optimal to encode with $\bin$ until $v$. However, the fact that $u$ is a dominating point for $\enc$ implies that $\bin(i, u) > \enc(i, u)$, which is absurd because $u < v$ and $\bin$ is optimal until $v$.
Therefore, both edges use the same encode. Assume that it is $\enc$ (the case for $\bin$ is symmetric).

The fact that $v$ belongs to the optimal solution $\mathcal{P}^*$ means that if we split the edge into two (or more) pieces, we \emph{cannot} decrease the cost, i.e., $\enc(i, v) \leq \enc(i, k) + \bin(k, v) + F$, $\forall i \leq k \leq v$.
Since $\enc$ is point-wise, we have $\enc(i, v) - \enc(i, k) = \enc(k, v)$ and thus, by imposing $k = u$, we obtain (1) $\enc(u, v) \leq \bin(u, v) + F$.
Vice versa, the fact that $u$ is a dominating point for $\enc$ means that from $u$ to $v$ the cost is \emph{higher} if we keep the same encoder, i.e., $\enc(i, v) \geq \enc(i, u) + \bin(u, v) + F$. Again, by exploiting the fact that $\enc$ is point-wise, we have (2) $\enc(u, v) \geq \bin(u, v) + F$.
Conditions (1) and (2) together imply that it must be $\enc(u, v) = \bin(u, v) + F$, thus we have no change in the cost of $\mathcal{P}^*$ by performing the exchange, which contradicts our assumption that $\mathcal{P}$ was not optimal.
\end{proof}

\begin{figure}[t]
\removelatexerror
\centering
\scalebox{1.0}{
    \input{optimal_partitioning.tex}
}
\caption{The \code{optimal\_partitioning} algorithm.
\label{alg:opt_split}}
\end{figure}

\begin{figure}[h]
\removelatexerror
\centering
\scalebox{1.0}{
    \input{update.tex}
}
\caption{The \code{update} algorithm.
\label{alg:update}}
\end{figure}

\begin{figure}[h]
\removelatexerror
\centering
\scalebox{1.0}{
    \input{close.tex}
}
\caption{The \code{close} algorithm.
\label{alg:close}}
\end{figure}

We are now left to present a detailed algorithm that computes $\mathcal{P}$, i.e., that iteratively finds all the dominating points of $S$ according to Equation~\ref{eq:recursive}.
We argue that the function \textsf{optimal\_partitioning} coded in Fig.~\ref{alg:opt_split} does the job.
Before proving that the algorithm is correct, let us explain the meaning of the variables used in the pseudocode.

Call $\ell$ the last added position to $\mathcal{P}$.
Variables $\var{i}$ and $\var{j}$ keep track of the positions of the points dominating the interval $S[\ell, \var{k})$ for, respectively, $\bin$ and $\enc$ encoders.
Likewise, $\var{max} = g(\var{i})$ and $\var{min} = g(\var{j})$ according to Definition~\ref{defn:gain}, with $\var{g}$ being the gain at step $\var{k}$.

\begin{lem}\label{lem:correctness}
The algorithm in Fig.~\ref{alg:opt_split} is \emph{correct}.
\end{lem}
\begin{proof}
We want to show that the array $\mathcal{P}$ returned by the function \textsf{optimal\_partitioning} contains all the positions of the dominating points, as recursively described by Equation~\ref{eq:recursive}.
We proceed by induction on the elements of $\mathcal{P}$.

The main loop in lines 6-17:
\begin{enumerate}
\item computes the gain $\var{g}$ at step $\var{k}$ (line 7);
\item updates the variables $\var{i}$, $\var{max}$ (lines 9-10) and $\var{j}$, $\var{min}$ (lines 14-15);
\item add new positions to $\mathcal{P}$ (lines 11-12 and 16-17).
\end{enumerate}
Correctness of 1) and 2) is immediate: the crucial point to explain is the third.

The \textsf{if} statements in lines 11 and 16 check whether positions $\var{i}$ and $\var{j}$ are dominating $[\ell, \var{k})$, i.e., whether $S[\var{i}]$ and $S[\var{j}]$ satisfy Definition~\ref{defn:dominating_point}.
Since the \textsf{if} statements are symmetric, we proved the correctness of the first one for non-decreasing values of $\var{g}$ (line 11).

We first check whether the $\var{min}$ gain, as seen so far, is sufficient to be the one of a dominating point for $\enc$ as required by Equation~\ref{eq:dominating_point}. At the beginning of the algorithm, the current interval starts at $\var{i} = 0$ and $T = F$, therefore $g(0) = 0$ in Equation~\ref{eq:dominating_point} and the test $\var{min} < -T$ is correct.
If $\var{min} < -T$ is true, then we also check if we have a sufficiently large increase in gain at the current step $\var{k}$ with respect to the previously seen $\var{min}$ gain according to Condition~\ref{eq:cond_1}. 
Again, it is immediate to see that the test $\var{min} - \var{g} < -2F$ checks such condition and, therefore, it is correct.
If both previous conditions are satisfied, then $j$ is the position of the dominating point for $\enc$ in the first interval $S[0, \var{k})$ by Definition~\ref{defn:dominating_point}.
If so, we can execute the \textsf{update} code, shown in Fig.~\ref{alg:update}, which adds $\var{j}$ to $\mathcal{P}$ and sets $T$ to $2F$ according to Definition~\ref{defn:dominating_point}.
Moreover, it updates the gain $\var{g}$ to maintain the invariant that its value is always relative to the current interval, which now begins at position $\var{j}$. In fact: since we have seen an increase of $2F$ bits, the $\var{max}$ gain in $S[\var{j}, \var{k})$ must be the current gain $\var{g}$, whereas the $\var{min}$ gain is 0 because $\var{g}$ is non-decreasing.
Thus, the first point is computed correctly.

Now, assume that we have added $h$ points to $\mathcal{P}$ and that the last added is for encoder $\enc$.
We want to show that the next point will be dominating for encoder $\bin$.
As explained before, whenever we add a dominating point for $\enc$ to $\mathcal{P}$, it means that we have seen an increase of $2F$ bits with respect to the last added position, i.e., position $\var{k}+1$ satisfies Equation~\ref{eq:dominating_point} for encoder $\bin$. Therefore the $(h+1)$-th point added to $\mathcal{P}$ will be dominating for $\bin$.

To conclude, we have to explain what happens at the end of the algorithm. Refer to the \textsf{close} function, coded in Fig.~\ref{alg:close}.
Lines 2-3 (4-5) check Condition~\ref{eq:cond_2}:
if successful, then $\var{max}$ ($\var{min}$) is the next dominating point for $\bin$ ($\enc$) and, since compressors must alternate each other, we close the encoding of the sequence with the other compressor in lines 6-9, that is $\enc$ ($\bin$);
if both unsuccessful, i.e., no dominating point is found, then it means that the remaining part of the sequence should not be cut and, thus, encoded with a single compressor in lines 6-9.
\end{proof}

In conclusion, since we consider each element of $S$ once and use a constant number of variables, Lemma~\ref{lem:optimality} and~\ref{lem:correctness} imply the following result.

\begin{thm}
A sorted integer sequence of size $n$ can be partitioned \emph{optimally} in $\Theta(n)$ time and $O(1)$ space, whenever its partitions are represented with a \emph{point-wise} encoder and \emph{characteristic bit-vectors}.
\end{thm}

\section{Experimental evaluation}\label{sec:experiments}

The aim of this section is the one of measuring the space improvement, indexing time and query processing speed of indexes that are optimally-partitioned by the algorithm described in Section~\ref{sec:optimal_splitting}.

\parag{Datasets}
We perform our experiments on the following standard datasets, whose statistics are summarized in Table~\ref{tab:datasets_stats}.
\begin{itemize}
\item {\gov} is the TREC 2004 Terabyte Track test collection, consisting in roughly $25$ million .gov sites crawled in early 2004. Documents are truncated to 256KB.
\item {\clue} is the ClueWeb 2009 TREC Category B test collection, consisting in roughly $50$ million English web pages crawled between January and February 2009.
\item {\cc} is an English subset of the freely available news from CommonCrawl\footnote{\url{http://commoncrawl.org/2016/10/news-dataset-available}}, consisting of articles crawled from $09/01/16$ to $30/03/18$.
\end{itemize}

\begin{table}[t]
    \centering
    \caption{Basic statistics for the tested collections.}
    \scalebox{1.0}{
        \begin{tabular}{l r r r}
\toprule

& {\gov} & {\clue}  & {\cc} \\

\midrule

Documents
  & \num{24622347}
  & \num{50131015}
  & \num{43530315}
  \\

Terms
  & \num{35636425}
  & \num{92094694}
  & \num{43844574}
  \\

Postings
  & \num{5742630292}
  & \num{15857983641}
  & \num{20150335440}  
  \\
  
\bottomrule
\end{tabular}
    }
\label{tab:datasets_stats}
\end{table}

\noindent
Identifiers were assigned to documents (docIDs) by following the lexicographic order of their URLs~\cite{2007:silvestri}.

\begin{table*}
    \centering
    \caption{The performance of the Variable-Byte family.
    }
    \scalebox{1.0}{
    \begin{tabular}{
                          l
                          c@{\hspace{6pt}}
                          c@{\hspace{6pt}}
                          c@{\hspace{6pt}}
                          c
                          c
                          c@{\hspace{6pt}}
                          c@{\hspace{6pt}}
                          c@{\hspace{6pt}}
                          c
                          c
                          c@{\hspace{6pt}}
                          c@{\hspace{6pt}}
                          c@{\hspace{6pt}}
                          c
                          }
\toprule
    
& \multicolumn{4}{c@{}@{}@{}@{}@{}}{\gov}
&
& \multicolumn{4}{c@{}@{}@{}@{}@{}}{\clue}
&
& \multicolumn{4}{c@{}@{}@{}@{}@{}}{\cc}
\\

\cmidrule(lr){2-5}
\cmidrule(lr){7-10}
\cmidrule(lr){12-15}

& {{\docs}}
& {{\freqs}}
& {\textsf{building}}
& {\textsf{AND query}}
&
& {{\docs}}
& {{\freqs}}
& {\textsf{building}}
& {\textsf{AND query}}
&
& {{\docs}}
& {{\freqs}}
& {\textsf{building}}
& {\textsf{AND query}}
\\

& {\textsf{[bpi]}}
& {\textsf{[bpi]}}
& {\textsf{[min]}}
& {\textsf{[ms]}}
&
& {\textsf{[bpi]}}
& {\textsf{[bpi]}}
& {\textsf{[min]}}
& {\textsf{[ms]}}
&
& {\textsf{[bpi]}}
& {\textsf{[bpi]}}
& {\textsf{[min]}}
& {\textsf{[ms]}}
\\

\midrule

\textsf{Varint-GB}
& $11.15$
& $9.77$
& $10.60$
& $0.88$
&
& $11.43$
& $9.80$
& $46.50$
& $5.32$
&
& $11.12$
& $10.01$
& $58.40$
& $7.38$
\\

\textsf{Varint-G8IU}
& $10.43$
& $9.00$
& $18.00$
& $0.84$
&
& $10.84$
& $8.99$
& $65.80$
& $5.10$
&
& $10.23$
& {\pp}$8.93$
& $60.60$
& $6.93$
\\

\textsf{Masked-VByte}
& {\pp}$9.53$
& $8.02$
& $10.50$
& $0.90$
&
& {\pp}$9.91$
& $8.01$
& $45.50$
& $5.52$
&
& {\pp}$9.42$
& {\pp}$8.00$
& $60.40$
& $7.06$
\\

\textsf{Stream-VByte}
& $11.15$
& $9.77$
& $10.60$
& $0.86$
&
& $11.43$
& $9.80$
& $46.50$
& $5.30$
&
& $11.12$
& $10.01$
& $58.40$
& $7.06$
\\

\bottomrule
\end{tabular}

    }
    \label{tab:VB_family}
\end{table*}

\parag{Experimental setting and methodology}
All the experiments were run on a machine with 4 Intel i7-4790K CPUs (8 threads) clocked at 4.00 GHz and with 32 GB of RAM DDR3, running Linux 4.13.0 (Ubuntu 17.10), 64 bits.
The implementation of our partitioned indexes is in standard C++14 and it is a flexible template library allowing \emph{any} point-wise encoder to be used, provided that its interface exposes a method to compute the cost in bits of a single integer in constant time.
We based our implementation on the popular \textsf{ds2i}\footnote{\url{https://github.com/ot/ds2i}} project.
The source code, available at \url{https://github.com/jermp/opt_vbyte}
to favour further research and reproducibility of results,
was compiled with \textsf{gcc} 7.2.0 using the highest optimization setting (compilation flags \texttt{-march=native} and \texttt{-O3}).

To test the building time of the indexes we measure the time needed to perform the whole end-to-end process, that is: (1) fetch the inverted lists from disk to main memory; (2) encode them in main memory; (3)
save the whole index data structure back to a file on disk.
Since the process is mostly I/O bound, we make sure to avoid disk caching effects by clearing the disk cache before building the indexes.

To test the query processing speed of the indexes, we memory map the index data structures on disk and compute boolean conjunctions over a set of random queries drawn from \textsf{TREC-05} and \textsf{TREC-06} Efficiency Track topics.
We repeat each experiment three times to smooth fluctuations in the measurements and report the average. The query algorithm runs on a single core and timings are reported in milliseconds per query.

In all the experiments, we used the value $F = 64$ bits for partitioning the inverted lists for both $\vb$
and partitioned Elias-Fano (henceforth, \textsf{PEF}).

\parag{Organization of the experiments}
Since we adopt \textsf{VByte} as example point-wise encoder, the next section compares the performance of all the encoders in the \textsf{VByte} family in order to choose the most convenient for the subsequent experiments.
Then, we measure the benefits of applying our optimization algorithm on the chosen \textsf{VByte} encoder, by comparing the corresponding partitioned index against the un-partitioned counterpart.
Finally, we compare our proposal with many other inverted index representations.

\subsection{The Variable-Byte family}
Several $\vb$ variants have been proposed in the literature, each with a different stream organization.
We now discuss them and inspect their performance.

By assuming that the largest represented integer fits into 4 bytes, two bits are sufficient to describe the proper number of bytes needed to represent an integer.
In this way, groups of four integers require one control byte that has to be read once as a header information.
This optimization was introduced in Google's Varint-GB~\cite{2009:dean} and reduces the probability of a branch misprediction which, in turn, leads to higher instruction throughput.
As already mentioned in Section~\ref{sec:intro}, working with byte-aligned codes also opens the possibility of exploiting the parallelism of SIMD (Single Instruction Multiple Data) instructions of modern processors to further enhance the decoding speed.
This is the line of research taken by the recent proposals that we overview below.

Varint-G8IU~\cite{2011:stepanov.gangolli.ea} uses a similar idea to the one of Varint-GB but it fixes the number of compressed bytes rather than the number of integers: one control byte is used to describe a variable number of integers in a data segment of exactly 8 bytes, therefore each group can contain between two and eight compressed integers.
Masked-VByte~\cite{2015:plaisance.kurz.ea} directly works on the original $\vb$ format. The decoder first gathers the most significant bits of consecutive bytes using a dedicated SIMD instruction. Then, using previously-built lookup tables and a shuffle instruction, the data bytes are permuted to obtain the original integers.
Stream-VByte~\cite{2018:lemire.kurz.ea}, instead, separates the encoding of the control bytes from the data bytes, by writing them into separate streams.
This organization permits to decode multiple control bytes simultaneously and, therefore, reduce branch mispredictions that can stop the CPU pipeline execution when decoding the data stream.

 \begin{table*}
    \centering
    \caption{Space in average number of bits (\textsf{bpi}) per document ({\docs}) and frequency ({\freqs}).}
    \scalebox{1.0}{
    \begin{tabular}{
                          l
                          r@{\hspace{1pt}}
                          r@{\hspace{5pt}}
                          r@{\hspace{1pt}}
                          r
                          c
                          r@{\hspace{1pt}}
                          r@{\hspace{5pt}}
                          r@{\hspace{1pt}}
                          r
                          c
                          r@{\hspace{1pt}}
                          r@{\hspace{5pt}}
                          r@{\hspace{1pt}}
                          r
}
\toprule

& \multicolumn{4}{c@{}@{}@{}@{}@{}@{}}{\gov}
&
& \multicolumn{4}{c@{}@{}@{}@{}@{}@{}}{\clue}
&
& \multicolumn{4}{c@{}@{}@{}@{}@{}@{}}{\cc}
\\

\cmidrule(lr){2-5}
\cmidrule(lr){7-10}
\cmidrule(lr){12-15}

& \multicolumn{2}{@{}c@{}}{{\docs}}
& \multicolumn{2}{@{}c@{}}{{\freqs}}
&
& \multicolumn{2}{@{}c@{}}{{\docs}}
& \multicolumn{2}{@{}c@{}}{{\freqs}}
&
& \multicolumn{2}{@{}c@{}}{{\docs}}
& \multicolumn{2}{@{}c@{}}{{\freqs}}
\\

& \multicolumn{2}{@{}c@{}}{\textsf{[bpi]}}
& \multicolumn{2}{@{}c@{}}{\textsf{[bpi]}}
&
& \multicolumn{2}{@{}c@{}}{\textsf{[bpi]}}
& \multicolumn{2}{@{}c@{}}{\textsf{[bpi]}}
&
& \multicolumn{2}{@{}c@{}}{\textsf{[bpi]}}
& \multicolumn{2}{@{}c@{}}{\textsf{[bpi]}}
\\

\midrule

\textsf{VByte}
& $9.53$ & \color{DarkGray}{$(+95.7\%)$}
& $8.02$ & \color{DarkGray}{$(+163.9\%)$}
&
& $9.90$ & \color{DarkGray}{$(+51.5\%)$}
& $8.01$ & \color{DarkGray}{$(+222.4\%)$}
&
& $9.42$ & \color{DarkGray}{$(+37.4\%)$}
& $8.00$ & \color{DarkGray}{$(+234.8\%)$}
\\

\textsf{VByte unif.}
& $5.41$ & \color{DarkGray}{$(+11.1\%)$}
& $3.31$ & \color{DarkGray}{$(+8.9\%)$}
&
& $7.37$ & \color{DarkGray}{$(+12.7\%)$}
& $2.69$ & \color{DarkGray}{$(+8.5\%)$}
&
& $7.27$ & \color{DarkGray}{$(+6.1\%)$}
& $2.55$ & \color{DarkGray}{$(+6.5\%)$}
\\

\textsf{VByte $\epsilon$-opt.}
& $4.93$ & \color{DarkGray}{$(+1.2\%)$}
& $3.05$ & \color{DarkGray}{$(+0.5\%)$}
&
& $6.66$ & \color{DarkGray}{$(+1.8\%)$}
& $2.50$ & \color{DarkGray}{$(+0.7\%)$}
&
& $6.92$ & \color{DarkGray}{$(+1.0\%)$}
& $2.41$ & \color{DarkGray}{$(+1.0\%)$}
\\

\textsf{VByte opt.}
& ${4.87}$ &
& ${3.04}$ &
&
& ${6.54}$ &
& ${2.48}$ &
&
& ${6.85}$ &
& ${2.39}$ &
\\

\bottomrule
\end{tabular}
    }
\label{tab:VB_partitioned_unpartitioned.space}
\end{table*}

\parag{The performance}
To help us in deciding which $\vb$ encoder to choose for our subsequent analysis, we consider the Table~\ref{tab:VB_family}.
The data reported in the table illustrates how different $\vb$ stream organizations actually impact index space.
Since \textsf{Varint-GB} and \textsf{Stream-VByte} take exactly the same space, given that \textsf{Stream-VByte} stores the very same control and data bytes but concatenated in separate streams, in the following we refer to both versions as \textsf{Varint-GB}.
As we can see, the original $\vb$ format (referred to as \textsf{Masked-VByte} in Table~\ref{tab:VB_family} because it uses this algorithm to perform decoding) is the most space-efficient among the family.
This is no surprise given the distribution plotted in Fig.~\ref{fig:postings_distr}: it means that the majority of the encoded $d$-gaps falls in the interval $[0,2^7)$, otherwise the compression ratio of $\vb$ would have been worse than the one of \textsf{Varint-GB} and \textsf{Varint-G8IU}.
As an example, consider the sequence of $d$-gaps $\langle 132, 233, 246, 178 \rangle$. $\vb$ uses 8 bytes to represent such sequence, whereas \textsf{Varint-GB} uses 1 control byte and 4 data bytes, thus 5 bytes in total. When all values are in $[0,2^7)$, $\vb$ uses 4 bytes instead of 5 as needed by \textsf{Varint-GB}.
For this reason, the space usage for \textsf{Varint-GB} and \textsf{Varint-G8IU} is larger than the one of $\vb$: it is $16 \div 18\%$ larger for \textsf{Varint-GB}; $10\%$ larger for \textsf{Varint-G8IU}.
The control byte of \textsf{Varint-G8IU} stores a bit for every of the 8 data bytes: a 0 bit means that the corresponding byte completes a decoded integer, whereas a 1 bit means that the byte is part of a decoded integer or it is wasted. Thus, \textsf{Varint-G8IU} compress worse than plain $\vb$ due to the wasted bytes.
Finally notice that \textsf{Varint-GB} is slightly worse than \textsf{Varint-G8IU} because it uses 10 bits per integer instead of 9 for all integers in $[0, 2^8)$. In fact, the difference between these two encoders in less than 1 bit on the tested datasets.

The speed of the encoders is actually very similar for all alternatives. We used the \textsf{TREC-05} querylog to compute boolean conjunctions. The spread between the fastest (\textsf{Varint-G8IU}) and the slowest alternative (\textsf{Masked-VByte}) is as little as $6 \div 10 \%$. The same holds true for the building of the indexes where, as expected, the plain $\vb$ is the fastest and \textsf{Varint-G8IU} is slower (by up to $40\%$ on {\gov} and {\clue}).

In conclusion, for the reasons discussed above, i.e., better space occupancy, better index building time and competitive speed, we adopt the original $\vb$ stream organization and the \textsf{Masked-VByte} algorithm by Plaisance, Kurz and Lemire~\cite{2015:plaisance.kurz.ea} to perform sequential decoding.

\subsection{Optimized Variable-Byte indexes}
In this section, we evaluate the impact of our solution by comparing the optimally-partitioned $\vb$ indexes against the un-partitioned indexes and the ones obtained by using other partitioning strategies, like uniform and the {$\epsilon$-optimal} based on dynamic programming (see Section~\ref{subsec:opt_part}).

As a high-level overview, the result of the comparison shows that our optimally-partitioned $\vb$ indexes are
$2\times$ smaller than the original, un-partitioned, counterparts; 
can be built $2\times$ faster without resorting on dynamic programming and offer the strongest guarantee, i.e., an exact solution rather than an approximation;
despite the significant space savings, these are as fast as the original $\vb$ indexes.

\parag{Index space}
Table~\ref{tab:VB_partitioned_unpartitioned.space} shows the results concerning the space of the indexes.
Compared to the case of un-partitioned indexes, we observe gains ranging from $37\%$ up to $235\%$, with a net factor of $2\times$ improvement with respect to the original $\vb$ format.

For the {uniform} partitioning we used partitions of 128 integers, for both documents and frequencies.
As we can see, this simple strategy already produces significant space savings: it is $43\%$, $26\%$ and $23\%$ better on the {\docs} sequences for $\gov$, $\clue$ and $\cc$ respectively; $59\%$, $66\%$ and $70\%$ better on the {\freqs} sequences.
This is because most $d$-gaps are actually very small but \emph{any} un-partitioned $\vb$ encoder needs at least 8 bits per $d$-gap.
In fact, notice how the average bits per integer on both {\docs} and {\freqs} becomes sensibly less than 8.

We recall that the $\epsilon$-optimal algorithm based on dynamic programming and reviewed in Section~\ref{subsec:opt_part}, was originally proposed for Elias-Fano~\cite{2014:ottaviano.venturini}, whose cost in bits can be computed in $O(1)$: we adapt the dynamic programming recurrence in order to use it for $\vb$ too.
As approximation parameters we used the same values as used in the experiments of the original work~\cite{2014:ottaviano.venturini}, i.e., we set $\epsilon_1 = 0.03$ and $\epsilon_2 = 0.3$.
The computed approximation could be possibly made large by enlarging such parameters, whereas our algorithm finds an exact solution.
However, we notice that the approximation is good and our optimal solution is only slightly better (by $1 \div 1.8\%$).
Compared to {uniform}, the {optimal} partitioning pays off: indeed it produces a further saving of $10\%$ on average, thus confirming the need for an optimization algorithm.

\parag{Index building time}
Although the un-partitioned variant would be the fastest to build in internal memory because the inverted lists are compressed in the same pass in which these are read from disk, the serialization of the data structure
imposes a considerable overhead because of the high memory footprint of the un-partitioned index. Notice how this factor becomes dramatic for the (larger) dataset $\clue$ and $\cc$, resulting in an end-to-end overhead of $50 \div 70\%$.
Because of this, also observe that there is no appreciable difference between the indexing time of the simple uniform
strategy and the optimal one.
Despite the linear-time complexity as soon as $\epsilon$ is constant, the $\epsilon$-{optimal} solution has a noticeable CPU cost due to the high constant factor, as we motivated in Section~\ref{subsec:opt_part}.
The {optimal} solutions has instead low constant factors and, as a result, is faster than the dynamic programming approach by more than $2.6\times$ on average on both $\gov$ and $\clue$; by $1.7\times$ on {\cc}.

\begin{table}
    \centering
    \caption{Index building timings in minutes.}
    \scalebox{1.0}{
    \begin{tabular}{
                          l@{\hspace{3pt}}
                          r@{\hspace{1pt}}
                          r@{\hspace{1pt}}
                          c@{\hspace{3pt}}
                          r@{\hspace{1pt}}
                          r@{\hspace{1pt}}
                          c@{\hspace{3pt}}
                          r@{\hspace{1pt}}
                          r@{\hspace{1pt}}
                          }
\toprule
    
& \multicolumn{2}{l@{}@{}}{\gov}
&
& \multicolumn{2}{l@{}@{}}{\clue}
&
& \multicolumn{2}{l@{}@{}}{\cc}
\\

\midrule
\textsf{VByte} & $10.1$ & \color{DarkGray}{$(-4\%)$} & & $43.3$ & \color{DarkGray}{$(+52\%)$} & & $60.4$ & \color{DarkGray}{$(+70\%)$} \\\textsf{VByte unif.} & $11.3$ & \color{DarkGray}{$(+8\%)$} & & $29.3$ & \color{DarkGray}{$(+3\%)$} & & $34.9$ & \color{DarkGray}{$(-2\%)$} \\\textsf{VByte $\epsilon$-opt.} & $26.7$ & \color{DarkGray}{$(+154\%)$} & & $72.3$ & \color{DarkGray}{$(+154\%)$} & & $59.8$ & \color{DarkGray}{$(+68\%)$} \\

\textsf{VByte opt.}
& ${10.5}$ &
&
& ${28.5}$ &
&
& ${35.5}$ &
\\

\bottomrule
\end{tabular}
    }
    \label{tab:VB_partitioned_unpartitioned.building_timings}
\end{table}

\begin{table}
    \centering
    \caption{Timings for AND queries in milliseconds.}
    \scalebox{1.0}{
    \begin{tabular}{
                          l@{\hspace{5pt}}
                          l@{\hspace{5pt}}
                          r@{\hspace{2pt}}
                          r@{\hspace{10pt}}
                          r@{\hspace{2pt}}
                          r@{\hspace{10pt}}
                          r@{\hspace{2pt}}
                          r@{\hspace{2pt}}
                          }
\toprule

&
& \multicolumn{2}{l@{}@{}}{\gov}
& \multicolumn{2}{l@{}@{}}{\clue}
& \multicolumn{2}{l@{}@{}}{\cc}
\\

\midrule

\multirow{4}{*}{\rotatebox[origin=c]{90}{{\fontsize{3mm}{3mm}\selectfont \sf TREC-05}}}
&  \textsf{VByte} & $0.90$ & \color{DarkGray}{$(+1\%)$} & $5.56$ & \color{DarkGray}{$(-3\%)$} & $7.06$ & \color{DarkGray}{$(+10\%)$} \\&  \textsf{VByte unif.} & $0.94$ & \color{DarkGray}{$(+5\%)$} & $5.90$ & \color{DarkGray}{$(+3\%)$} & $7.20$ & \color{DarkGray}{$(+13\%)$} \\&  \textsf{VByte $\epsilon$-opt.} & $0.92$ & \color{DarkGray}{$(+3\%)$} & $5.89$ & \color{DarkGray}{$(+3\%)$} & $6.52$ & \color{DarkGray}{$(+2\%)$} \\
&
\textsf{VByte opt.}
& ${0.89}$ &
& ${5.70}$ &
& ${6.39}$ &
\\

\midrule

\multirow{4}{*}{\rotatebox[origin=c]{90}{{\fontsize{3mm}{3mm}\selectfont \sf TREC-06}}}
&  \textsf{VByte} & $2.12$ & \color{DarkGray}{$(+0\%)$} & $8.35$ & \color{DarkGray}{$(-7\%)$} & $9.36$ & \color{DarkGray}{$(+12\%)$} \\&  \textsf{VByte unif.} & $2.22$ & \color{DarkGray}{$(+5\%)$} & $9.02$ & \color{DarkGray}{$(+1\%)$} & $9.58$ & \color{DarkGray}{$(+14\%)$} \\&  \textsf{VByte $\epsilon$-opt.} & $2.24$ & \color{DarkGray}{$(+6\%)$} & $9.17$ & \color{DarkGray}{$(+2\%)$} & $8.56$ & \color{DarkGray}{$(+2\%)$} \\
&
\textsf{VByte opt.}
& ${2.12}$ &
& ${8.96}$ &
& ${8.38}$ &
\\

\bottomrule
\end{tabular}
    }
    \label{tab:VB_partitioned_unpartitioned.query_timings}
\end{table}

%
%

\begin{table*}[t]
    \centering
    \caption{Space in average number of bits (\textsf{bpi}) per document ({\docs}) and frequency ({\freqs}).}
    \scalebox{1.0}{
    \begin{tabular}{
                          l
                          r@{\hspace{1pt}}
                          r@{\hspace{5pt}}
                          r@{\hspace{1pt}}
                          r
                          c
                          r@{\hspace{1pt}}
                          r@{\hspace{5pt}}
                          r@{\hspace{1pt}}
                          r
                          c
                          r@{\hspace{1pt}}
                          r@{\hspace{5pt}}
                          r@{\hspace{1pt}}
                          r
}
\toprule

& \multicolumn{4}{c@{}@{}@{}@{}@{}@{}}{\gov}
&
& \multicolumn{4}{c@{}@{}@{}@{}@{}@{}}{\clue}
&
& \multicolumn{4}{c@{}@{}@{}@{}@{}@{}}{\cc}
\\

\cmidrule(lr){2-5}
\cmidrule(lr){7-10}
\cmidrule(lr){12-15}

& \multicolumn{2}{@{}c@{}}{{\docs}}
& \multicolumn{2}{@{}c@{}}{{\freqs}}
&
& \multicolumn{2}{@{}c@{}}{{\docs}}
& \multicolumn{2}{@{}c@{}}{{\freqs}}
&
& \multicolumn{2}{@{}c@{}}{{\docs}}
& \multicolumn{2}{@{}c@{}}{{\freqs}}
\\

& \multicolumn{2}{@{}c@{}}{\textsf{[bpi]}}
& \multicolumn{2}{@{}c@{}}{\textsf{[bpi]}}
&
& \multicolumn{2}{@{}c@{}}{\textsf{[bpi]}}
& \multicolumn{2}{@{}c@{}}{\textsf{[bpi]}}
&
& \multicolumn{2}{@{}c@{}}{\textsf{[bpi]}}
& \multicolumn{2}{@{}c@{}}{\textsf{[bpi]}}
\\

\midrule

\textsf{PEF $\epsilon$-opt.}
& $4.10$ & \color{DarkGray}{$(-15.7\%)$}
& $2.38$ & \color{DarkGray}{$(-21.8\%)$}
&
& $5.85$ & \color{DarkGray}{$(-10.6\%)$}
& $2.20$ & \color{DarkGray}{$(-11.6\%)$}
&
& $5.84$ & \color{DarkGray}{$(-14.8\%)$}
& $2.18$ & \color{DarkGray}{$(-8.9\%)$}
\\

\textsf{OptPFD}
& $4.48$ & \color{DarkGray}{$(-8.0\%)$}
& $2.38$ & \color{DarkGray}{$(-21.8\%)$}
&
& $6.18$ & \color{DarkGray}{$(-5.4\%)$}
& $2.41$ & \color{DarkGray}{$(-2.9\%)$}
&
& $6.41$ & \color{DarkGray}{$(-6.5\%)$}
& $2.53$ & \color{DarkGray}{$(+5.9\%)$}
\\

\textsf{BIC}
& $3.80$ & \color{DarkGray}{$(-22.0\%)$}
& $2.14$ & \color{DarkGray}{$(-29.5\%)$}
&
& $5.15$ & \color{DarkGray}{$(-21.3\%)$}
& $1.87$ & \color{DarkGray}{$(-24.8\%)$}
&
& $5.37$ & \color{DarkGray}{$(-21.7\%)$}
& $1.98$ & \color{DarkGray}{$(-17.3\%)$}
\\

\textsf{ANS}
& $3.96$ & \color{DarkGray}{$(-18.7\%)$}
& $1.85$ & \color{DarkGray}{$(-39.0\%)$}
&
& $5.36$ & \color{DarkGray}{$(-18.0\%)$}
& $1.94$ & \color{DarkGray}{$(-21.9\%)$}
&
& $5.76$ & \color{DarkGray}{$(-16.0\%)$}
& $2.01$ & \color{DarkGray}{$(-15.8\%)$}
\\

\textsf{QMX}
& $6.00$ & \color{DarkGray}{$(+23.3\%)$}
& $3.37$ & \color{DarkGray}{$(+10.8\%)$}
&
& $8.01$ & \color{DarkGray}{$(+22.6\%)$}
& $3.75$ & \color{DarkGray}{$(+51.2\%)$}
&
& $7.31$ & \color{DarkGray}{$(+6.6\%)$}
& $3.72$ & \color{DarkGray}{$(+55.5\%)$}
\\

\textsf{VByte opt.}
& ${4.87}$ &
& ${3.04}$ &
&
& ${6.54}$ &
& ${2.48}$ &
&
& ${6.85}$ &
& ${2.39}$ &
\\

\bottomrule
\end{tabular}
    }
\label{tab:VB_comparison.space}
\end{table*}

\begin{figure}[t]
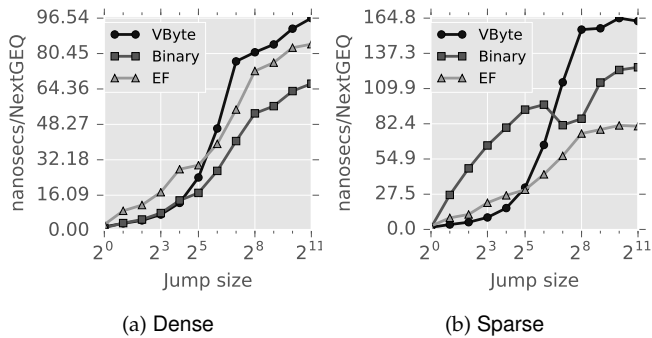

    \hspace*{-0.2cm}    
    \subfloat[\textsf{Dense}]{
    \includegraphics[scale=0.7]{{{plots/gov2.next_geqs.dense.bw}}}
    \label{fig:next_geqs.dense}
    }
    \hspace*{-0.5cm}
    \subfloat[\textsf{Sparse}]{
    \includegraphics[scale=0.7]{{{plots/gov2.next_geqs.sparse.bw}}}
     \label{fig:next_geqs.sparse}
    }
    \caption{Average nanoseconds spent per \textsf{NextGEQ} query for (a) \textsf{Dense} and (b) \textsf{Sparse} sequences of one million integers.
}
     \label{fig:next_geqs}
\end{figure}

\begin{figure}[t]
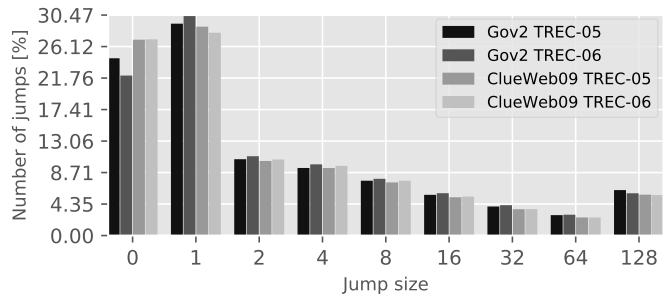

	\hspace*{-0.2cm}
    \includegraphics[scale=0.75]{{{plots/jumps.bw}}}
    \caption{When the difference between two consecutively accessed positions is $d$, \textsf{NextGEQ} is said to make a jump of size $d$.
The distribution of the jump sizes is divided into buckets of exponential size: all sizes between $2^{d-1}$ and $2^d$ belong to bucket $d$.
The plot shows the jumps distribution, in percentage, for the querylogs used in the experiments, when performing AND queries.}
    \label{fig:jumps}
\end{figure}

\parag{Index speed}
Table~\ref{tab:VB_partitioned_unpartitioned.query_timings} illustrates the results.
The striking result of the experiment is that, despite the significant space reduction ($2\times$ improvement, see Table~\ref{tab:VB_partitioned_unpartitioned.space}), the partitioned indexes are as fast as the un-partitioned ones on all datasets.

Therefore, it is important to provide a careful explanation of such result.
The answer is provided by understanding the plots in Fig.~\ref{fig:next_geqs}, along with the ones in Fig.~\ref{fig:postings_distr} and~\ref{fig:jumps}, that we generate for {\gov} and {\clue} since we obtained even better timing results for {\cc} (see Table~\ref{tab:VB_partitioned_unpartitioned.query_timings}).
In particular, Fig.~\ref{fig:next_geqs} illustrates the average nanoseconds spent per \textsf{NextGEQ} (Next Greater-than or EQual-to) query by $\vb$, the binary vector representation and Elias-Fano (\textsf{EF}).
The $\textsf{NextGEQ}_t(x)$ query returns the \emph{smallest} integer $z \geq x$ from the inverted list of the term $t$ and it is the core operation needed to perform fast intersection~\cite{EBDT2018}.
The timings reported in Fig.~\ref{fig:next_geqs} are relative to a sequence of one million integers and with an average gap between the integers of: (a) 2.5, as a \emph{dense} case and (b) 1850 as a \emph{sparse} case.
These values mimic the ones for the $\gov$ dataset, that are 2.13 and 1852 respectively. The ones for the $\clue$ dataset are 2.14 and 963, thus the plots have a similar shape.

As the dense case illustrates, the binary vector representation is as fast as \textsf{VByte} for all jumps of entity less then or equal to 8, and becomes actually faster for longer jumps.
Moreover, the distribution of the jump sizes plotted in Fig.~\ref{fig:jumps} indicates that, whenever executing AND queries, the number of jumps of size less than 16 accounts for $\approx$90$\%$ of the jumps performed by \textsf{NextGEQ}.
Furthermore, the distribution plotted in Fig.~\ref{fig:postings_distr} tells us that the majority of blocks are actually encoded with their characteristic bit-vector, thus explaining why the partitioned indexes exhibit no penalty against the un-partitioned counterparts.

However, $\vb$ tends to be slower on longer jumps because of its block-wise organization: since a posting list is split into blocks of 128 postings that are encoded separately, a block must be completely decoded even for accessing a single integer, which is not uncommon for boolean conjunctions.
Moreover, since $d$-gaps values are encoded, we need to access the elements by a linear scan of the block after decoding in order to compute their prefix sums.
When the accessed elements per block are very few, even using SIMD instructions to perform decoding results in a slower query execution.
Conversely and as expected, the binary vector representation is inefficient for the \emph{sparse} regions since potentially many bits need to be scanned to perform a query, but still faster than $\vb$ whenever the jump size becomes larger than 64 because it allows skipping over the bit stream by keeping samples of the bit set positions.

%

\subsection{Overall comparison}
In this section we compare the optimally-partitioned $\vb$ indexes against several competitors:
\begin{itemize}
\item the $\epsilon$-optimal partitioned Elias-Fano method (\textsf{PEF}) by Ottaviano and Venturini~\cite{2014:ottaviano.venturini};
\item the Binary Interpolative coding (\textsf{BIC}) by Moffat and Stuiver~\cite{2000:moffat.stuiver};
\item the optimized PForDelta (\textsf{OptPDF}) by Yan \emph{et al.}~\cite{2009:yan.ding.ea};
\item the \textsf{ANS}-based index by Moffat and Petri~\cite{ANS2};
\item the \textsf{QMX} mechanism by Trotman~\cite{2014:trotman}.
\end{itemize}

For all competitors, we used the {C++} code from the original authors, compiled with \textsf{gcc} 7.2.0 using the highest optimization setting
as we did for our own code to ensure a fair comparison.


\parag{Index space and building time}
Table~\ref{tab:VB_comparison.space} shows the results concerning the space of the indexes.
Clearly, the space usage of the $\vb$ {optimal} indexes is higher than the one of the bit-aligned encoders: this was expected since $\vb$ is byte-aligned.
However, the important result is that \emph{its space is not so high as it used to be before}.
In fact, comparing the results reported in Table~\ref{tab:VB_partitioned_unpartitioned.space} with the ones in Table~\ref{tab:VB_comparison.space}, we see that, without partitioning, $\vb$ was $172\%$ larger than \textsf{PEF} and $194\%$ larger than \textsf{BIC} on $\gov$; $123\%$ larger than \textsf{PEF} and $154\%$ larger than \textsf{BIC} on $\clue$;
$117\%$ larger than \textsf{PEF} and $137\%$ larger than \textsf{BIC} on $\cc$.
Now, thanks to our optimization strategy, this gap is reduced to $20\%$ on average.
In particular, notice that it is less than $11\%$ larger than \textsf{PEF} on the {\docs} sequences of {\clue}, while it is generally less effective on the {\freqs} sequences. This is because the within-document frequencies are made up of integers smaller than docIDs.
Very similar considerations hold for the other alternatives, such as \textsf{OptPFD} and \textsf{ANS}.
In particular, we notice that on the $\clue$ dataset, the difference between $\vb$ {optimal} and \textsf{OptPDF} is very small (only $4\%$ overall);
\textsf{BIC} is (as usual) generally better than other methods on the both {\docs} and {\freqs};
the byte-aligned \textsf{QMX} is, instead, significantly larger, by $19 \div 31\%$.

\return
We now consider the time needed to build the indexes.
Refer to Table~\ref{tab:VB_comparison.building_timings}.
As already noted in the previous subsection, the dynamic programming approach used for \textsf{PEF} imposes a severe penalty with respect to \textsf{VByte optimal} of $4\times$ on average.
The penalty is due to not only the difference in speed between dynamic programming and the algorithm devised in Section~\ref{sec:optimal_splitting}, but also to the fact that Elias-Fano, being bit-aligned, is slower to encode with respect to $\vb$.
Except for the \textsf{ANS} indexes which are slower to build, by $33\%$ on average, because of the two-pass process of first collecting symbol occurrence counts and, then, encoding~\cite{ANS2}, the building timings for the other competitors are, instead, competitive:
our optimization algorithm only takes a couple of minutes more overall the whole building process.
Only \textsf{BIC} and \textsf{QMX} took less indexing time ($33\%$ faster on average on {\gov} and {\clue}, but only $16\%$ more on the largest {\cc} dataset).

\begin{table}
    \caption{Index building timings in minutes.}
    \centering
    \scalebox{1.0}{
    \begin{tabular}{
                          l@{\hspace{3pt}}
                          r@{\hspace{1pt}}
                          r@{\hspace{1pt}}
                          c@{\hspace{3pt}}
                          r@{\hspace{1pt}}
                          r@{\hspace{1pt}}
                          c@{\hspace{3pt}}
                          r@{\hspace{1pt}}
                          r@{\hspace{1pt}}
                          }
\toprule

& \multicolumn{2}{l@{}@{}}{\gov}
&
& \multicolumn{2}{l@{}@{}}{\clue}
&
& \multicolumn{2}{l@{}@{}}{\cc}
\\

\midrule

\textsf{PEF $\epsilon$-opt.} & $41.3$ & \color{DarkGray}{$(+293\%)$} & & $125.5$ & \color{DarkGray}{$(+340\%)$} & & $85.2$ & \color{DarkGray}{$(+140\%)$} \\\textsf{OptPFD} & $8.2$ & \color{DarkGray}{$(-22\%)$} & & $25.8$ & \color{DarkGray}{$(-9\%)$} & & $36.7$ & \color{DarkGray}{$(+3\%)$} \\\textsf{BIC} & $7.0$ & \color{DarkGray}{$(-33\%)$} & & $20.5$ & \color{DarkGray}{$(-28\%)$} & & $28.3$ & \color{DarkGray}{$(-20\%)$} \\\textsf{ANS} & $12.6$ & \color{DarkGray}{$(+20\%)$} & & $35.7$ & \color{DarkGray}{$(+25\%)$} & & $55.1$ & \color{DarkGray}{$(+55\%)$} \\\textsf{QMX} & $7.0$ & \color{DarkGray}{$(-33\%)$} & & $18.0$ & \color{DarkGray}{$(-37\%)$} & & $31.1$ & \color{DarkGray}{$(-12\%)$} \\

\textsf{VByte opt.}
& ${10.5}$ &
&
& ${28.5}$ &
&
& ${35.5}$ &
\\

\bottomrule
\end{tabular}
    }
\label{tab:VB_comparison.building_timings}
\end{table}

\parag{Index speed}
Table~\ref{tab:VB_comparison.query_timings} shows the query processing speed of the indexes.
Compared to \textsf{PEF}, the results are indeed very similar to the ones obtained by Ottaviano and Venturini~\cite{2014:ottaviano.venturini}, i.e., there is only a marginal gap between the speed of \textsf{PEF} and $\vb$ when computing boolean conjunctions.
The reason has to be found, again, in the plot illustrated in Fig.~\ref{fig:next_geqs.sparse}. As we can see, for all the jump sizes less than 32, $\vb$ is $2\times$ faster than Elias-Fano, while this advantage vanishes for the longer jumps thanks to the powerful skipping abilities of Elias-Fano~\cite{EBDT2018,2014:ottaviano.venturini,2013:vigna}.
However, we know that this advantage is shrunk because jumps larger than 32 are not very frequent on the tested query logs, as depicted by the distribution of Fig.~\ref{fig:jumps}.

Compared to the other approaches, we can see significant gains with respect to \textsf{OptPDF} (by $40\%$ on $\gov$ and $21\%$ on $\clue$), \textsf{BIC} and \textsf{ANS} ($4\times$ faster on average) and only a slight penalty with respect to \textsf{QMX} (by $7 \div 10\%$) on $\clue$.
On the largest {\cc} dataset, our proposal is consistently the fastest approach.

\begin{table}
    \caption{Timings for AND queries in milliseconds.}
    \centering
    \scalebox{1.0}{
    \begin{tabular}{
                          l@{\hspace{5pt}}
                          l@{\hspace{5pt}}
                          r@{\hspace{2pt}}
                          r@{\hspace{10pt}}
                          r@{\hspace{2pt}}
                          r@{\hspace{10pt}}
                          r@{\hspace{2pt}}
                          r@{\hspace{2pt}}
                          }
\toprule

&
& \multicolumn{2}{l@{}@{}}{\gov}
& \multicolumn{2}{l@{}@{}}{\clue}
& \multicolumn{2}{l@{}@{}}{\cc}
\\

\midrule

\multirow{6}{*}{\rotatebox[origin=c]{90}{{\fontsize{3mm}{3mm}\selectfont \sf TREC-05}}}
&  \textsf{PEF $\epsilon$-opt.} & $0.98$ & \color{DarkGray}{$(+10\%)$} & $5.87$ & \color{DarkGray}{$(+3\%)$} & $8.64$ & \color{DarkGray}{$(+35\%)$} \\&  \textsf{OptPFD} & $1.28$ & \color{DarkGray}{$(+43\%)$} & $8.04$ & \color{DarkGray}{$(+41\%)$} & $8.92$ & \color{DarkGray}{$(+40\%)$} \\&  \textsf{BIC} & $4.14$ & \color{DarkGray}{$(+364\%)$} & $25.42$ & \color{DarkGray}{$(+346\%)$} & $63.50$ & \color{DarkGray}{$(+894\%)$} \\&  \textsf{ANS} & $4.21$ & \color{DarkGray}{$(+372\%)$} & $25.98$ & \color{DarkGray}{$(+356\%)$} & $27.49$ & \color{DarkGray}{$(+330\%)$} \\&  \textsf{QMX} & $0.88$ & \color{DarkGray}{$(-1\%)$} & $5.30$ & \color{DarkGray}{$(-7\%)$} & $7.09$ & \color{DarkGray}{$(+11\%)$} \\
&
\textsf{VByte opt.}
& ${0.89}$ &
& ${5.70}$ &
& ${6.39}$ &
\\

\midrule

\multirow{6}{*}{\rotatebox[origin=c]{90}{{\fontsize{3mm}{3mm}\selectfont \sf TREC-06}}}
&  \textsf{PEF $\epsilon$-opt.} & $2.19$ & \color{DarkGray}{$(+4\%)$} & $9.59$ & \color{DarkGray}{$(+7\%)$} & $11.77$ & \color{DarkGray}{$(+40\%)$} \\&  \textsf{OptPFD} & $3.00$ & \color{DarkGray}{$(+42\%)$} & $11.95$ & \color{DarkGray}{$(+33\%)$} & $11.73$ & \color{DarkGray}{$(+40\%)$} \\&  \textsf{BIC} & $9.93$ & \color{DarkGray}{$(+369\%)$} & $37.87$ & \color{DarkGray}{$(+322\%)$} & $81.52$ & \color{DarkGray}{$(+873\%)$} \\&  \textsf{ANS} & $9.48$ & \color{DarkGray}{$(+348\%)$} & $38.07$ & \color{DarkGray}{$(+325\%)$} & $35.48$ & \color{DarkGray}{$(+323\%)$} \\&  \textsf{QMX} & $2.11$ & \color{DarkGray}{$(-1\%)$} & $8.07$ & \color{DarkGray}{$(-10\%)$} & $9.44$ & \color{DarkGray}{$(+13\%)$} \\
&
\textsf{VByte opt.}
& ${2.12}$ &
& ${8.96}$ &
& ${8.38}$ &
\\

\bottomrule
\end{tabular}
    }
\label{tab:VB_comparison.query_timings}
\end{table}

%

\section{Conclusions}\label{sec:conclusions}
We have presented an optimization algorithm for point-wise encoders that splits a sorted integer sequence into variable-sized partitions to improve its compression and has a linear-time/constant-space complexity.
We have also proved that the algorithm is \emph{optimal}, i.e., it finds the partitioning that minimizes the space of the representation.
For point-wise encoders, this is sensibly better than approaches based on dynamic-programming on all aspects: time/space complexity and practical performance.

By applying our technique to the ubiquitous Variable-Byte encoding, we have exhibited a $2\times$-better compression ratio and build
optimally-partitioned indexes $2\times$ faster than the linear-time dynamic programming approach.
Despite the significant space savings, the partitioned representation does not introduce penalties at query processing time
compared to the un-partitioned case.

As a last note, we mention the possibility of introducing another encoder for \emph{representing the runs} of the posting lists. Obviously, a run of consecutive integers can be described with just the information stored in the first level of representation, i.e., the size of the run.
Although our framework can be extended to include this case, the algorithm and its analysis become much more complicated.
This additional complexity may not pay off, because space improved by less than 5\% on the tested datasets.

\bibliographystyle{IEEEtran}
\bibliography{IEEEabrv,bibliography}

\vspace{-1.5cm}
\begin{IEEEbiography}[{\includegraphics[width=1in,height=1.25in,clip,keepaspectratio]{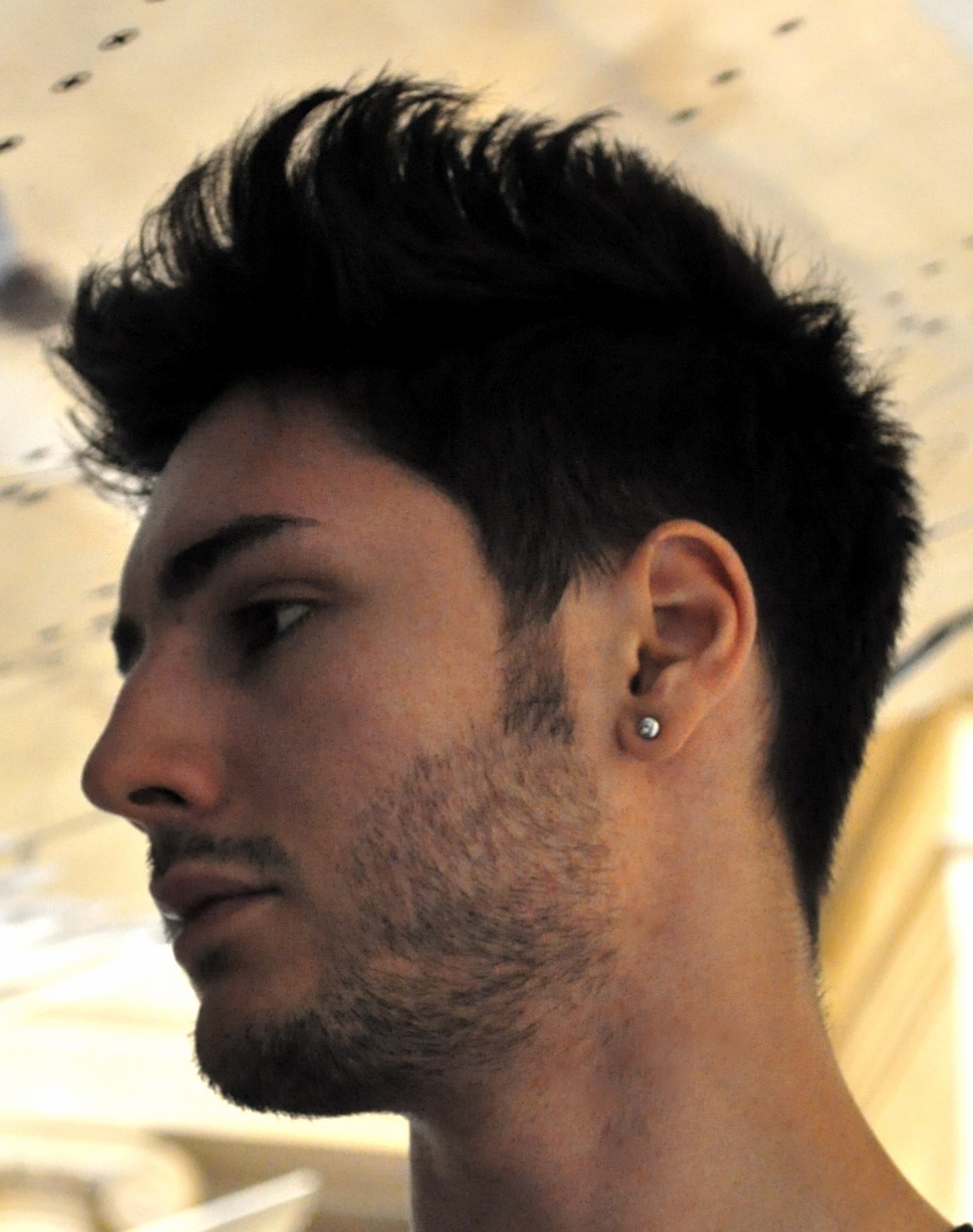}}]{Giulio Ermanno Pibiri}(\url{http://pages.di.unipi.it/pibiri}) received his Bachelor Degree in Computer Engineering from the University of Florence in 2012; Master Degree and Ph.D. in Computer Science from the University of Pisa in 2015 and 2019 respectively.
His research interests involve data compression algorithms for indexing massive datasets, data structures and information retrieval with focus on efficiency.
\end{IEEEbiography}

\vspace{-1.3cm}
\begin{IEEEbiography}
[{\includegraphics[width=1in,height=1.25in,clip,keepaspectratio]{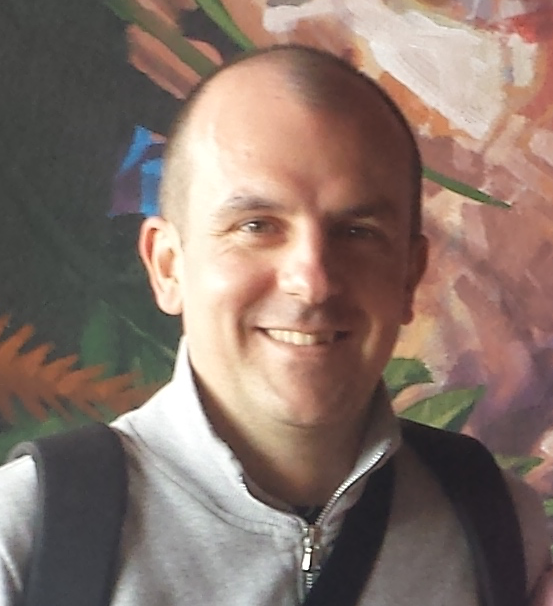}}] {Rossano Venturini}(\url{http://pages.di.unipi.it/rossano}) is Associate Professor of Computer Science at the University of Pisa. He received his Ph.D. from the University of Pisa in 2010. His research interests are mainly focused on the design and the analysis of algorithms and data structures with focus on indexing and searching large textual collections.
He won two Best Paper Awards at ACM SIGIR in 2014 and 2015.
\end{IEEEbiography} 

\end{document}